\pdfoutput=1

\documentclass[conference,compsoc]{IEEEtran}
\ifCLASSOPTIONcompsoc
  \usepackage[nocompress]{cite}
\else
  \usepackage{cite}
\fi
\ifCLASSINFOpdf
  \usepackage[pdftex]{graphicx}
\else
  \usepackage[dvips]{graphicx}
\fi
\usepackage{amssymb,amsmath,amsthm}
\usepackage{enumerate}

\newtheorem{theorem}{Theorem}[section]
\newtheorem{lemma}[theorem]{Lemma}

\usepackage{algorithmic}
\usepackage[lined, boxed, comments numbered, ruled, linesnumbered]{algorithm2e}
\ifCLASSOPTIONcompsoc
  \usepackage[caption=false,font=footnotesize,labelfont=sf,textfont=sf]{subfig}
\else
  \usepackage[caption=false,font=footnotesize]{subfig}
\fi
\begin{document}
%
\title{LGRASS: Linear Graph Spectral Sparsification for Final Task of the 3rd ACM-China International Parallel Computing Challenge
}


\author{\IEEEauthorblockN{Yuxuan Chen, Jiyan Qiu, Zidong Han, Chenhan Bai}
\IEEEauthorblockA{CNIC, Chinese Academy of Sciences and University of Chinese Academy of Sciences\\
Beijing, China \\
{\{chenyuxuan, qiujiyan, zdhan, chbai\}@cnic.cn}}
}    


%


\maketitle

\begin{abstract}
This paper presents our solution for optimization task of the 3rd ACM-China IPCC. By the complexity analysis, we identified three time-consuming subroutines of original algorithm: marking edges, computing pseudo inverse and sorting edges. These subroutines becomes the main performance bottleneck owing to their super-linear time complexity.
To address this, we proposed LGRASS, a linear graph spectral sparsification algorithm to run in strictly linear time. LGRASS takes advantage of spanning tree properties and efficient algorithms to optimize bottleneck subroutines. Furthermore, we crafted a parallel processing scheme for LGRASS to make full use of multi-processor hardware. Experiment shows that our proposed method fulfils the task in dozens of milliseconds on official test cases and keep its linearity as graph size scales up on random test cases.
\end{abstract}


%
\IEEEpeerreviewmaketitle

\section{Introduction}

The final task of the 3rd ACM-China IPCC is to  optimize a program for graph spectral sparsification. In this competition, three test cases and source code of a program are provided for participants. These test cases are used for performance evaluation, which consist of 4K, 7K and 16K nodes  respectively.
Participants of this competition need to craft a program that runs as fast as possible while outputs the same result as provided program on designated hardware platform.

For this task, we first proposed LGRASS, a graph spectral sparsification algorithm which runs in strictly linear time. Then we crafted a parallel processing scheme to optimize the hot-spot subroutines of LGRASS. The details of the two-step improvements is discussed in the following sections.

\section{Motivation}

\begin{figure}[htbp]
    \centering
    \subfloat[Baseline]{
    \includegraphics[width=1.0in]{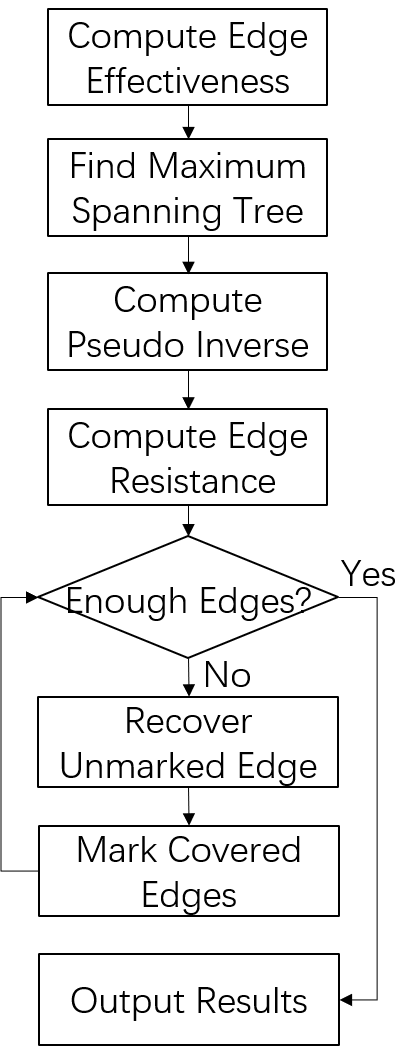}
    \label{figure:baseline}
    }
    \subfloat[Basic LGRASS]{
    \includegraphics[width=1.0in]{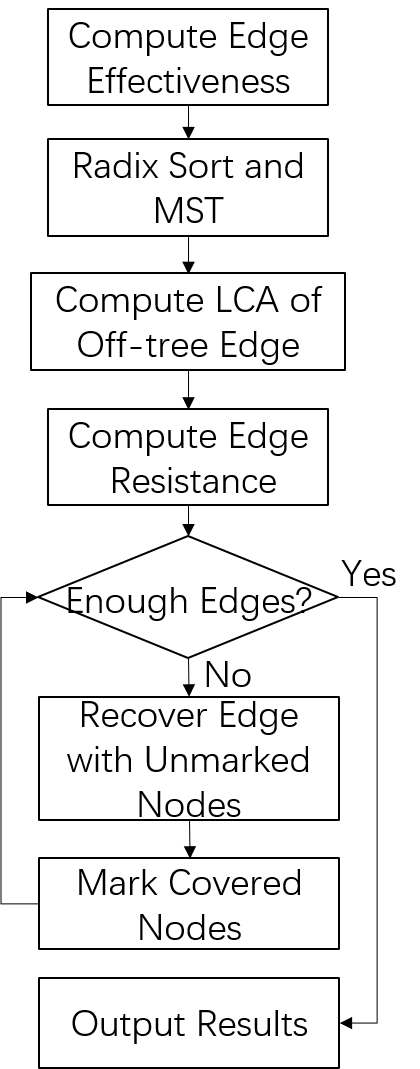}
    \label{figure:LGRASS}
    }
    \subfloat[Parallel LGRASS]{
    \includegraphics[width=1.0in]{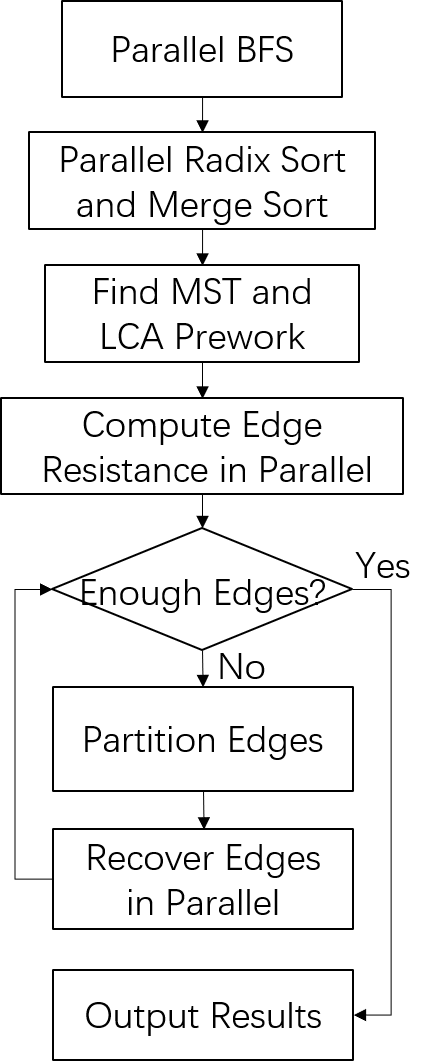}
    \label{figure:pLGRASS}
    }
    \caption{Program Procedures}
\end{figure}

\begin{table}[ht]
    \centering
    \caption{Time breakdown of baseline program}
    \begin{tabular}{|c|c|c|c|c|c|c|}
        \hline
         & EFF & MST & INV & RES & MARK & ALL\\
        \hline
        Case 1 & 10.7ms & 205.3ms & 10.1s & 33.1ms & 21.7min & 22.9min \\
        \hline
        Case 2 & 5.6ms & 183.6ms & 52.4s & 13.4ms & 24.6min & 25.5min\\
        \hline
    \end{tabular}
    \label{tab:baseline_breakdown}
\end{table}

To find hot-spot subroutines, we made a coarse-grained performance profiling for baseline program, whose procedure is demonstrated in Fig. \ref{figure:baseline}.
As is shown in Table \ref{tab:baseline_breakdown}, the running time of the baseline program, i.e. provided program, is dominated by its subroutine of marking edges. By the computational complexity analysis, there are three subroutines which costs  super-linear time: marking edges, computing pseudo inverse and sorting edges. These subroutines costs over 99.9\% time of the overall running time, which means their super-linear running time has become the main performance bottleneck.

Thus, we propose LGRASS, to improve the time complexity of these three subroutines and reduce their running time to be \emph{linear}, and then optimize it for multiprocessor systems. The procedures of basic version and parallel version of our proposed method are demonstrated in Fig. \ref{figure:LGRASS} and Fig. \ref{figure:pLGRASS} respectively. The novel techniques used by basic LGRASS is presented in Section \ref{sec:BasicLGRASS} and the parallel processing scheme of parallel LGRASS is depicted in Section \ref{sec:ParallelLGRASS}.

\section{The Design of LGRASS}
\label{sec:BasicLGRASS}
\subsection{Marking Edges in Linear Time}

The pseudocode of marking edges in baseline program is described in Algorithm \ref{alg:baseline_marking}. The running time of Algorithm \ref{alg:baseline_marking} is dominated by a three-level nested loop with time complexity $O(N^2L)$, where $N$ is the number of nodes and $L$ is the number of edges. We observe that marking edges is arduous, while checking the edge mark only needs \emph{one} memory access, which is in constant time.

\begin{algorithm}
\caption{Marking Edges in Baseline Program}
\label{alg:baseline_marking}
\KwIn{Edge $(u, v)$}

$lca$ = the lowest common ancestor of $u$ and $v$ \\
$beta$ = max(min(depth[$u$], depth[$v$]) - depth[$lca$], 1) \\
$S_1$ = the nodes covered by $u$ with distance $beta$ \\
$S_2$ = the nodes covered by $v$ with distance $beta$ \\
\For{$x \in S_1$}{
    \For{$y \in S_2$}{
        \For{$e \in E$} {
            \If{e = (x, y) or e = (y, x)}{
                Mark edge $e$
            }
        }
    }
}
\end{algorithm}

To mitigate the cost of marking edges, we start with the properties of \emph{marked edges}. An edge marked by $(u,v)$ must consist of two nodes covered by $u$ and $v$ respectively. By attaching edge marks to covered nodes, we can achieve $O(N)$ and $O(L)$ time complexity for marking edges and checking the edge mark respectively, which is described in Algorithm \ref{alg:linear_marking} and Algorithm \ref{alg:linear_checking}.

\begin{algorithm}
\caption{Marking Edges in Linear Time}
\label{alg:linear_marking}
\KwIn{Edge $e$ = $(u, v)$}

$lca$ = the lowest common ancestor of $u$ and $v$ \\
$beta$ = max(min(depth[$u$], depth[$v$]) - depth[$lca$], 1) \\
$S_1$ = the nodes covered by $u$ with distance $beta$ \\
$S_2$ = the nodes covered by $v$ with distance $beta$ \\
\For{$x \in S_1$}{
    $M_x$ = $M_x \cup \{ e_1 \}$
}
\For{$y \in S_2$}{
    $M_y$ = $M_y \cup \{ e_2 \}$
}
\end{algorithm}

\begin{algorithm}
\caption{Checking Edge Mark in Linear Time}
\label{alg:linear_checking}
\KwIn{Edge $(u, v)$}
\KwOut{Mark flag $flag$}

$flag$ = False \\
\For {$e \in E$ } {
    \If{$e_1 \in M_u$ and $e_2 \in M_v$}{
        $flag$ = True
    }
    \If{$e_1 \in M_v$ and $e_2 \in M_u$}{
        $flag$ = True
    }    
}
\end{algorithm}

By the analysis of edges in test cases,  we observe that the majority of off-tree edges are crossing edges, i.e. the edges that connect two nodes such that they do not have any ancestor and a descendant relationship between them. 
We then focus on the properties of crossing edges, giving Lemma \ref{lemma:rlca} and Lemma \ref{lemma:lca}.

\begin{lemma}
\label{lemma:rlca}
For edge (x, y) and crossing edge (u, v), if edge (x, y) is covered by edge (u, v), then LCA(x, y) = LCA(u, v).
\end{lemma}

\begin{proof}
Without loss of generality, assume $x$ is covered by $u$ and $y$ is covered by $v$.

Let $w$ = LCA($u$, $v$). Note that $(u, v)$ is a crossing edge, which means $w \neq u$, then depth[$u$] $>$ depth[$w$]. So $beta \leq $ depth[$u$] - depth[$w$]. By $w \neq u$, $u$ is in a subtree of $w$, denoted by $T_1$. Note that $x$ is covered by $u$ with distance $beta \leq $ depth[$u$] - depth[$w$], then $x$ is in $T_1$ or $x$ is $w$. 
Similarly, $y$ is in the subtree $T_2$ of $w$, or $y$ is $w$. Note that $T_1 \cap T_2$ = $\emptyset$, so LCA($x$, $y$) = $w$ = LCA($u$, $v$).
\end{proof}

\begin{lemma}
\label{lemma:lca}
For crossing edge (x, y) and crossing edge (u, v), if node x and node y are covered by node u or node v, and LCA(x, y) = LCA(u, v), then edge (x, y) is covered by edge (u, v).
\end{lemma}

\begin{proof}
We prove this lemma by the method of contradiction. Assume the edge $(x, y)$ is not covered by edge $(u, v)$. Without loss of generality, assume node $x$ and $y$ are both covered by node $u$.  

Let $w$ = LCA($x$, $y$) = LCA($u$, $v$). Then $beta$ = min(depth[$u$], depth[$v$]) - depth[$w$] $\leq$ depth[$u$] - depth[$w$].
Note that $(x, y)$ is a crossing edge and LCA($x$, $y$) = $w$, then max(dist($x$, $u$), dist($y$, $u$)) $>$ dist($w$, $u$) = depth[$u$] - depth[$w$] $\geq beta$. So at least one node between $x$ and $y$ is not covered by $u$, which contradicts with the assumption.
\end{proof}

The converse proposition of Lemma \ref{lemma:lca} is also true, of which the proof is omitted. Then we have Algorithm \ref{alg:cross_marking} and Algorithm \ref{alg:cross_checking} to process crossing-edge marking and checking.

\begin{algorithm}
\caption{Crossing-edge Marking}
\label{alg:cross_marking}
\KwIn{Crossing edge $e$ = $(u, v)$}

$lca$ = the lowest common ancestor of $u$ and $v$ \\
$beta$ = min(depth[$u$], depth[$v$]) - depth[$lca$] \\
$S_1$ = the nodes covered by $u$ with distance $beta$ \\
$S_2$ = the nodes covered by $v$ with distance $beta$ \\
\For{$x \in S_1$}{
    $M_{lca,x}$ = $M_{lca,x} \cup \{ e \}$
}
\For{$y \in S_2$}{
    $M_{lca,y}$ = $M_{lca,y} \cup \{ e \}$
}
\end{algorithm}

\begin{algorithm}
\caption{Checking Crossing-edge Mark}
\label{alg:cross_checking}
\KwIn{Edge $(u, v)$}
\KwOut{Mark flag $flag$}

$lca$ = the lowest common ancestor of $u$ and $v$ \\

\uIf{$u \neq lca$ and $v \neq lca$}{
    $flag$ = $M_{lca,u} \cap M_{lca,v}$ is not empty
}\Else{
    $flag$ = False \\
    \For{$e \in M_{lca,u} \cap M_{lca,v}$}{
        \If{(u, v) is covered by $e$}{
            $flag$ = True
        }
    }
}

\end{algorithm}

Although Algorithm \ref{alg:cross_checking} has no improvement on time complexity compared with Algorithm \ref{alg:linear_checking}, it can leverage the classic acceleration techniques for set operations, such as bitmap and SIMD, to reduce both running time and memory footprint.

As for non-crossing edges, the main obstacle is their aftereffects, which may withdraw marks of some crossing edges to effect more edges. So we designed an extra stage to recover non-crossing edges and deal with the complex aftereffects, which is described in Algorithm \ref{alg:noncross_fix}.

\begin{algorithm}
\caption{Recovering Non-crossing Edges}
\label{alg:noncross_fix}

\For{$e \in E$ }{
    $flipped$ = False \\
    \If{(e.isMarked and e.isEnforced) or \\ 
    (not e.isMarked and e.isWithdrawn)}{
        $flag$ = checking all-edge mark for $e$ \\
        \If{$flag \neq e.isMarked$ }{
            $flipped$ = True \\
            $e.isMarked$ = $flag$
        }
    }
    \If{(not e.isCrossing and not e.isMarked) \\ or $flipped$}{
        \For{$e' \in $ edges covered by $e$}{
            \uIf{$e.isMarked$}{
                $e'.isEnforced$ = True
            }\Else{
                $e'.isWithdrawn$ = True
            }
        }
    }
}
\end{algorithm}

\subsection{Computing Resistance in Linear Time}

The subroutine of computing resistance in baseline program is involved with operations on Laplacian matrix of minimal spanning tree. Considering the size of Laplacian matrix, the running time of operations on Laplacian matrix is at least quadratic. To achieve linear time complexity, we make use of the resistance computation subroutine of [1], which is proved to run in $O(L)$ time.

In our implementation, we made a slight change on LCA computation but greatly reduce its time consumption. For an off-tree edge $(u, v)$, we first check if $u$ and $v$ are in the same subtree of root. If so, compute the LCA of $u$ and $v$ as former procedure. Otherwise, the LCA of $u$ and $v$ must be root.

\subsection{Sorting Edges in Linear Time}

The baseline program utilizes stable sort function of STL library, whose time complexity is $O(L\log L)$. We notice that our data to sort is a bunch of non-negative double-precision floating-point number excluding NaN and infinity. By the IEEE standard [2], for 64-bit binary representations $a$ and $b$, if FP($a$) $<$ FP($b$), then INT($a$) $<$ INT($b$). That is, we can sort our data in an INT64 manner. 

We choose radix sort as our basic sort algorithm for its stability, linearity and simplicity. Each 64-bit number is divided into eight unsigned 8-bit integer keys, because 256 int-type bucket counters can be put into  exactly one page. With one round of scanning and eight rounds of relocation, the edges can be sorted in $O(L)$ time.

\section{Parallel Optimizations}
\label{sec:ParallelLGRASS}
\subsection{An Analysis of Hot-spot Subroutines}
Table \ref{tab:lgrass_breakdown} shows the time breakdown of basic LGRASS. 
By our analysis, the subroutine of marking edges consumes up to 43.5\% of overall running time, while the other subroutines except LCA computation consumes 9.6\%-25.1\% of overall running time each.

\begin{table}[ht]
    \centering
    \caption{Time breakdown of basic LGRASS}
    \begin{tabular}{|c|c|c|c|c|c|c|}
        \hline
         & EFF & MST & LCA & RES & MARK & ALL\\
        \hline
        Case 1 & 2.8ms & 4.3ms & 2.1ms & 3.3ms & 4.6ms & 17.1ms \\
        \hline
        Case 2 & 1.5ms & 2.4ms & 1.1ms & 1.5ms & 5.0ms & 11.5ms\\
        \hline
        Case 3 & 14.9ms & 14.1ms & 9.5ms & 10.4ms & 33.5ms & 82.4ms \\
        \hline
    \end{tabular}
    \label{tab:lgrass_breakdown}
\end{table}

Diving into these subroutines, we identified four parallel-optimization opportunities: edge marking, resistance computation, BFS (for computing effectiveness) and merge sort (for MST edge sort and resistance edge sort).
In the following subsections, we depict how we exploit these opportunities.

\subsection{Parallel Edge Marking}

By Lemma \ref{lemma:rlca}, the task of crossing-edge marking can be partitioned based on the LCA of off-tree edges. However, the majority of off-tree edges recognize root as their LCA, such that the workload of root-LCA subtask becomes much heavier than that of other subtasks. So we made a further partition based on which two subtrees those edges connect. The two-step partition is implemented with a carefully designed mapping function:

\begin{equation*} F(u, v) =
\begin{cases}
LCA(u, v) &\text{if $LCA(u, v) \neq root$} \\
N &\text{if $u = root$ or $v = root$} \\
N + 1 + \binom{S_1}{2} + S_2 &\text{otherwise}
\end{cases} \end{equation*}
where $N$ is the number of nodes, $S_1$ and $S_2$ are the maximal and minimal subtree index of edge $(u, v)$ respectively. And the node index and subtree index are both ranged from 0. After partition, each subtask can be independently processed with Algorithm \ref{alg:cross_marking} and Algorithm \ref{alg:cross_checking}. When all subtasks are done, Algorithm \ref{alg:noncross_fix} is still necessary to deal with the aftereffects of non-crossing edges. 
The comparison of edge marking procedure between basic LGRASS and parallel LGRASS is demonstrated in Fig. \ref{figure:parallel}.

Besides, we implemented a greedy scheduling algorithm for dynamic task dispatching, which outperforms static task dispatching when using more than 4 threads.

\subsection{Parallel Resistance Computation }

\begin{figure}[htbp]
    \centering
    \subfloat[Basic LGRASS]{
    \includegraphics[width=1.5in]{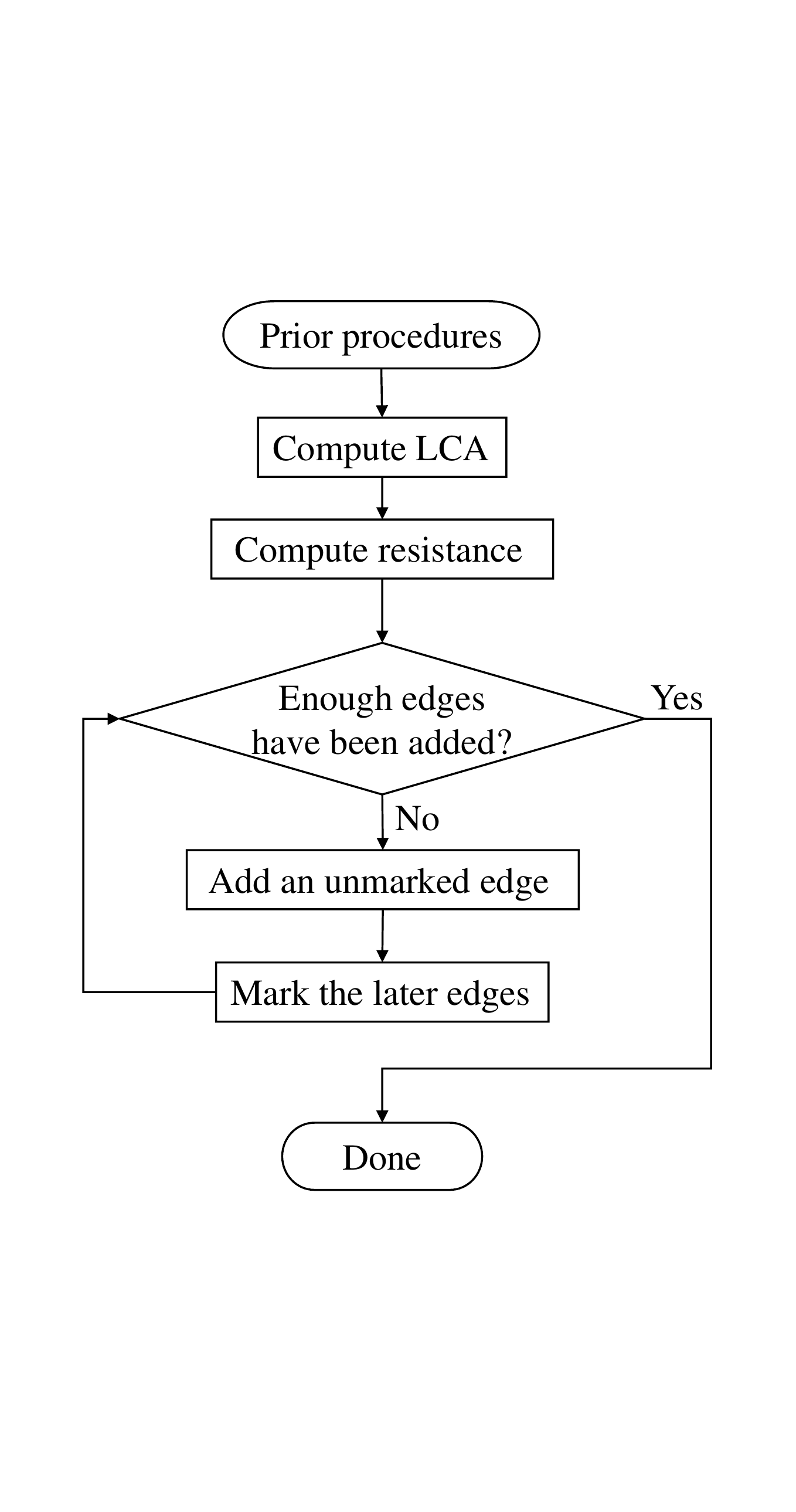}
    }
    \subfloat[Parallel LGRASS]{
    \includegraphics[width=1.5in]{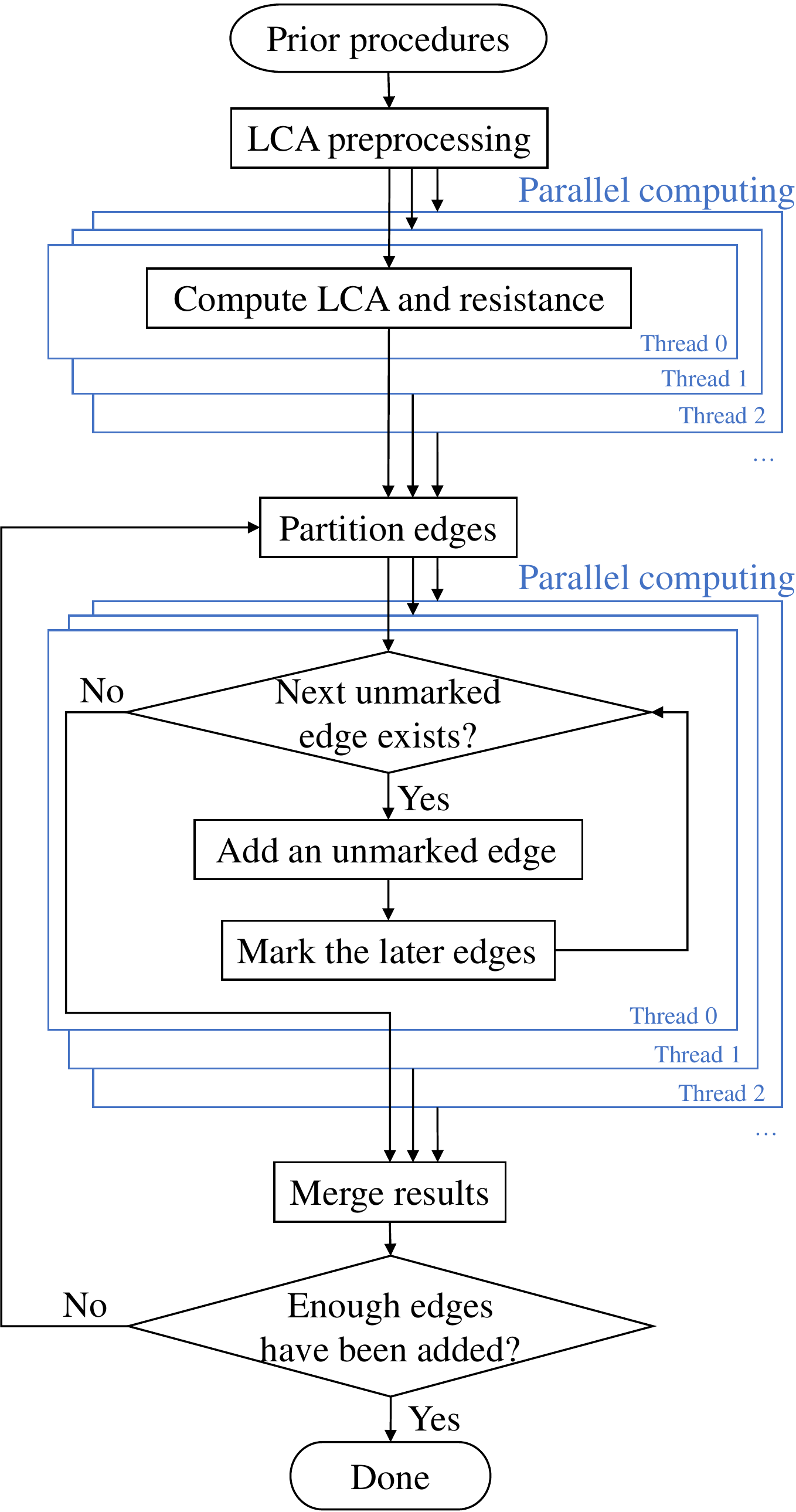}
    }
    \caption{Procedures of Edge Marking and Resistance Computation}
    \label{figure:parallel}
\end{figure}

To make full use of multiprocessors, the resistance computation task is uniformly partitioned for each thread. For the same purpose, part of LCA computation is offloaded to resistance computation subroutine, and an efficient online LCA algorithm is implemented replace the former one.
Fig. \ref{figure:parallel} shows the difference of LCA and resistance computation between the basic version and parallel version of LGRASS.

\subsection{Parallel Breadth-first Search}

By our profiling result, the running time of computing effectiveness is dominated by BFS. To alleviate this, we implemented a level-synchronous parallel BFS algorithm. By our implementation, each thread first obtains vertices of current layer from a concurrent queue, then searches for the vertices of next layer and pushes them into the concurrent queue for next layer. We employ an atomic variable to record the number of elements in a concurrent queue, which ensures the linearizability of concurrent operations for multiple threads. Besides, we use atomic operations on distance array in relaxed memory order to prevent a vertex from being added to the queue too many times.

We also took several tricks for better performance: relaxed memory order for mutex variables, hand-crafted assembly code for compare-and-set operations and local-to-global merging for layer aggregation. 
Fig. \ref{figure:pbfs} is an example of our parallel BFS scheme.

\begin{figure}[htbp]
    \centering
    \includegraphics[width=0.8\linewidth]{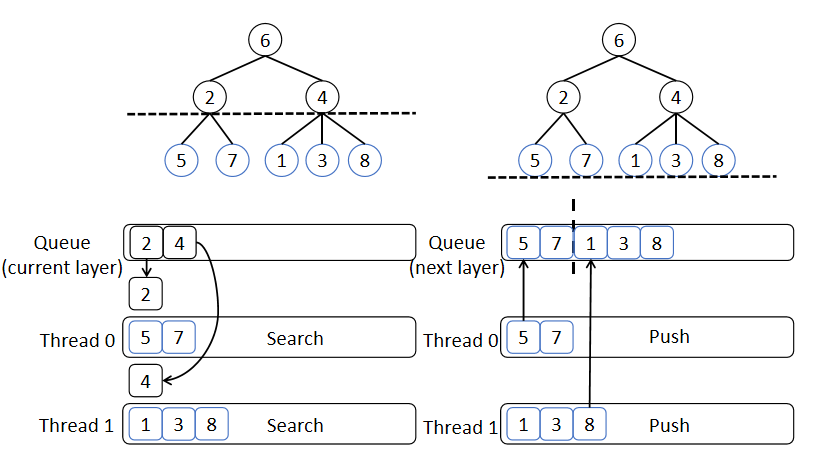}
    \caption{Parallel Breadth-first Search}
    \label{figure:pbfs}
\end{figure}

\subsection{Parallel Merge Sort}

Considering the inherent intractability of sequential radix sort, we take merge sort as the main framework of our parallel sort and regard radix sort as its submodule. 

By our design, each thread first sorts a block of the data by radix sort, then merge the blocks in parallel. If there are $P$ threads for merge sort, the thread with the longest running time need to perform up to $(2-\frac{1}{P})L$ comparisons. 

We observe that only a small fraction of sorted data is accessed by later procedure, so we leave the final merger to later procedure
and reduce the maximum number of comparisons down to $(1-\frac{1}{P})L$ in merge sort.

Furthermore, when sorting edges by resistance, we can merge only top-$K$ edges, where $K$ depends on the number of edges for partition. Then each thread only needs to perform no more than $(\lfloor \log_2{P} \rfloor - 1)K$ times comparisons in merge sort. The procedure is shown in Fig. \ref{figure:pmergesort}.

\begin{figure}[htbp]
    \centering
    \includegraphics[width=0.8\linewidth]{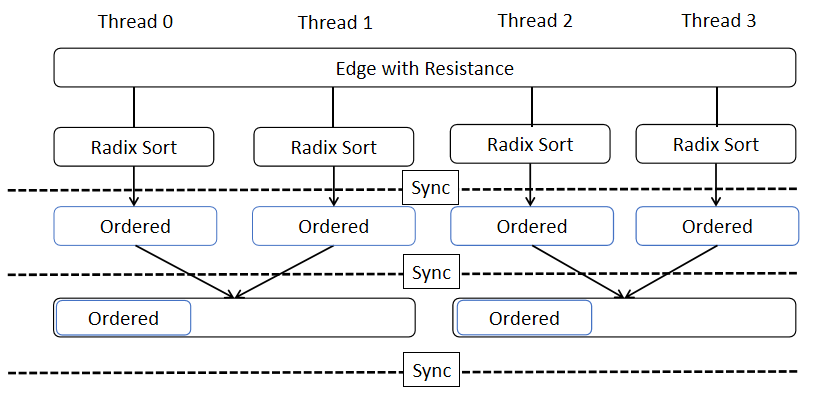}
    \caption{Parallel Merge Sort of Edges with Resistance}
    \label{figure:pmergesort}
\end{figure}




\section{Evaluation}

\begin{table}[ht]
    \centering
    \caption{Execution time comparison}
    \begin{tabular}{|c|c|c|c|}
        \hline
         & Baseline & Basic LGRASS & Parallel LGRASS\\
        \hline
        Case 1 & 22.9min & 15.3ms & 6.4ms \\
        \hline
        Case 2 & 25.5min & 10.8ms & 5.0ms\\
        \hline
        Case 3 & 37.1h & 73.7ms & 24.0ms \\
        \hline
    \end{tabular}
    \label{tab:comparison}
\end{table}

We evaluated our method on hardware platform provided by Beijing Super Cloud Computing Center, which is specified by organizers. Table \ref{tab:comparison} demonstrates the comparison among baseline program, basic LGRASS and parallel LGRASS. Compared with baseline program, basic LGRASS can reduce the execution time by several orders of magnitude with only \emph{one} processor. By utilizing multiprocessors, parallel LGRASS gains up to 3.1x speedup against basic LGRASS. 
We also evaluated our proposed method on random test cases. As is shown in Fig. \ref{figure:linearity}, the using time of parallel LGRASS increases linearly as graph size scales up.

\begin{figure}[htbp]
    \centering
    \includegraphics[width=0.7\linewidth]{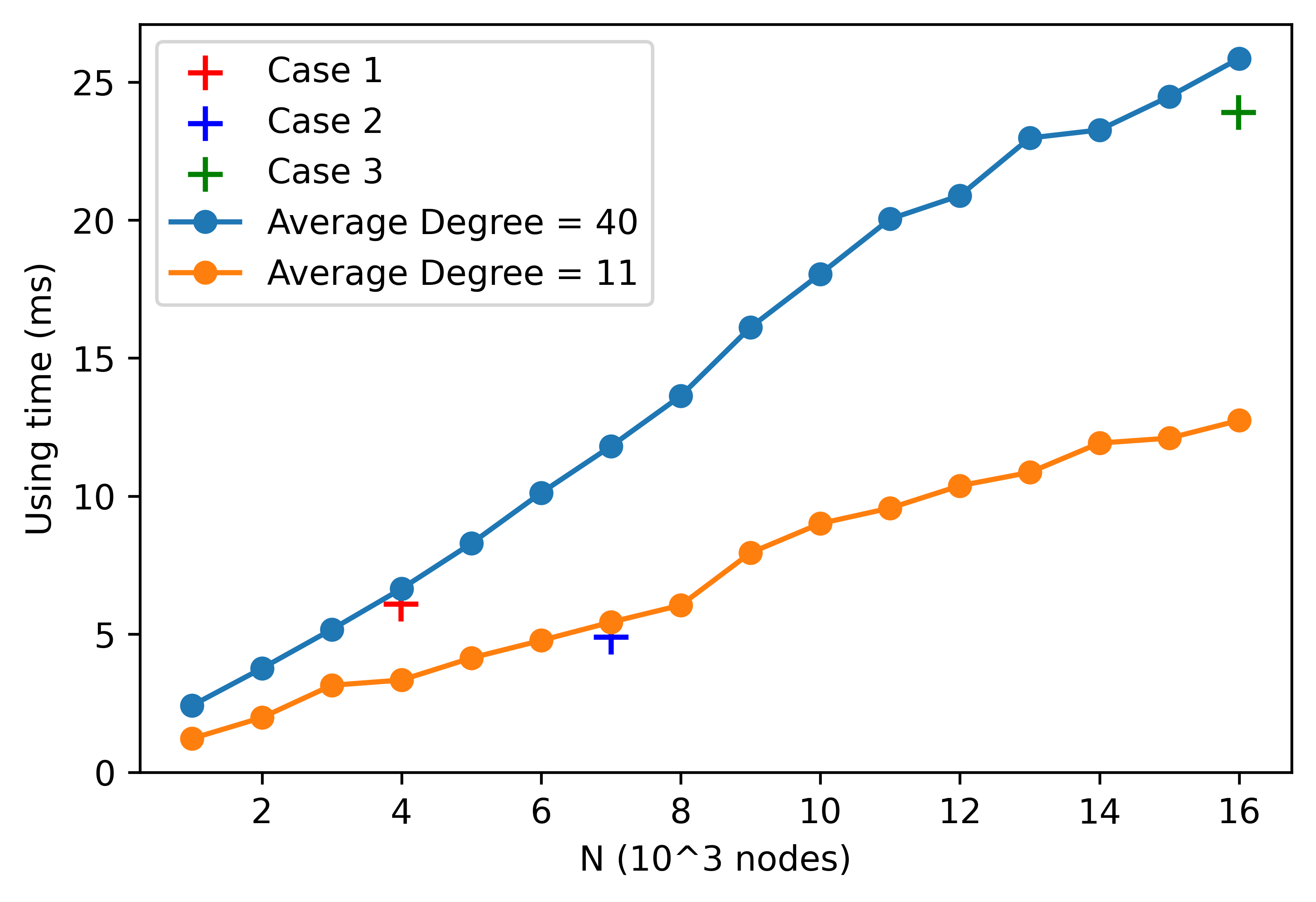}
    \caption{Using Time of Parallel LGRASS on Random Test Cases}
    \label{figure:linearity}
\end{figure}

\section{Conclusion}
In this paper, we present our solution to the final task of the 3rd ACM-China IPCC. We proposed LGRASS, a linear graph spectral sparsification algorithm, and optimized it for parallel processing. Evaluation result shows that our proposed method can fulfils the task in 6.4ms, 5.0ms and 24.0ms respectively on official test cases and keep its linearity as graph size scales up on random test cases.


\ifCLASSOPTIONcompsoc
  \section*{Acknowledgments}
\else
  \section*{Acknowledgment}
\fi

We thank our shepherd, Jian Zhang, for his feedback and encouragement.



%

\end{document}